\documentclass{svjour3}                     
\smartqed  
\usepackage{graphicx}
\usepackage{mathptmx}      
\usepackage{amsmath}
\usepackage{amssymb}
\usepackage{enumerate}
\usepackage{eucal}
\usepackage{latexsym}
\usepackage{rotating}
\newcommand{\A}{{\mathcal A}}
\newcommand{\C}{{\mathcal C}}
\newcommand{\N}{{\mathcal N}}
\renewcommand{\S}{{\mathcal S}}
\newcommand{\bbF}{{\mathbb F}}
\newcommand{\vu}{{\sf u}}
\newcommand{\vv}{{\sf v}}
\newcommand{\vx}{{\sf x}}
\newcommand{\vzero}{{\sf 0}}
\newcommand{\dist}{{\rm dist}}
\newcommand{\supp}{{\rm supp}}
\newcommand{\wt}{{\rm wt}}
\journalname{Des. Codes Cryptogr.}
\begin{document}

\title{List Decodability at Small Radii\thanks{The research of Y. M. Chee and S. Ling is supported in part by the National Research
Foundation of Singapore under Research Grant NRF-CRP2-2007-03. The research of Y. M. Chee
is also supported in part by the Nanyang Technological University under Research
grant M58110040.\\
The research of G. Ge is supported in part by the National Outstanding Youth Science Foundation
of China under Grant 10825103, National Natural Science Foundation of China under Grant
10771193, Specialized Research Fund for the Doctoral Program of Higher Education,
Program for New Century Excellent Talents in University, and Zhejiang Provincial Natural
Science Foundation of China under Grant D7080064. \\
The research of L. Ji is supported by NSFC under Grants 10701060, 10831002, and the
Qing Lan Project of Jiangsu province. \\
The research of J. Yin is supported by NSFC under Grants 10831002 and 10671140.}
}


\author{Yeow Meng Chee \and Gennian Ge \and Lijun Ji \and San Ling \and Jianxing Yin}


\institute{Y. M. Chee \and S. Ling \at
              Division of Mathematical Sciences, School of Physical \& Mathematical Sciences,
              Nanyang Technological University, Singapore 637371 \\
              \email{ymchee@ntu.edu.sg, lingsan@ntu.edu.sg}           
           \and
           G. Ge \at
           Department of Mathematics, Zhejiang University, Hangzhou 310027, Zhejiang, China \\
           \email{gnge@zju.edu.cn}
           \and
           L. Ji \and J. Yin \at
           Department of Mathematics, Suzhou University, Suzhou 215006, Jiangsu, China \\
           \email{jilijun@suda.edu.cn, jxyin@suda.edu.cn}
}

\date{Received: date / Accepted: date}

\maketitle

\begin{abstract}
$A'(n,d,e)$, the smallest $\ell$ for which every binary error-correcting code of length $n$ and 
minimum distance $d$ is decodable with a list of size $\ell$ up to radius $e$, is
determined for all $d\geq 2e-3$. As a result, $A'(n,d,e)$ is determined for all
$e\leq 4$, except for 42 values of $n$.
\end{abstract}

\section{Introduction}

If more than $d/2$ errors occur when using a binary error-correcting code
of minimum distance $d$, unambiguous decoding cannot always be guaranteed.
Instead of simply letting the decoding algorithm report a failure in this case,
{\em list decoding}, a notion introduced independently by
Elias \cite{Elias:1957} and Wozencraft \cite{Wozencraft:1958}, demands that the
decoding algorithm returns a small list of codewords that contains the transmitted
codeword, as a weak form of error recovery.

Early applications of list decoding include
\begin{enumerate}[(i)]
\item tighter analysis of error-probability and error-exponent of probabilistic channels
\cite{Elias:1957,Forney:1968,Wozencraft:1958},
\item derivation of the Elias-Bassalygo bound
\cite{Bassalygo:1965,Shannonetal:1967a,Shannonetal:1967b},
\item determination of channel capacities \cite{Ahlswede:1973}, and
\item error-correction under an adversarial model
\cite{Blinovsky:1986,Blinovsky:1997,Elias:1991,ZyablovPinsker:1982}.
\end{enumerate}
Renewed interest in list decoding in theoretical computer science
stemmed from the work of Goldreich and Levin
\cite{GoldreichLevin:1989}, and is largely due to the breakthrough discovery by Sudan
\cite{Sudan:1997} of the first efficient algorithm for list decoding a nontrivial code.
Sudan's work led to a multitude of new applications of list decoding in theoretical computer
science and became a powerful tool in complexity theory.

In this paper, we study the list decodability of error-correcting codes. More specifically,
we are interested in determining the smallest $\ell$ so that every error-correcting code
of length $n$ and
minimum distance $d$ can be list decoded with a list of size $\ell$ at a given radius $e$.
This parameter, denoted by $A'(n,d,e)$, was investigated by Elias \cite{Elias:1991}
and more recently by Guruswami and Sudan \cite{GuruswamiSudan:2001} as well
as Cassuto and Bruck \cite{CassutoBruck:2004}. These works
all focused on giving upper bounds on $A'(n,d,e)$. In contrast, our attention in this paper is on
determining the exact value of $A'(n,d,e)$ for specified $d$ and $e$.
More specifically, we determine the exact value of $A'(n,d,e)$ for all $e\leq 4$, 
except for 42 values
of $n$. A summary of the results obtained is provided in Table \ref{A''}.

\setlength\rotFPtop{0pt plus 1fil}
\begin{sidewaystable}
\renewcommand{\arraystretch}{1.5}
\centering
\caption{$A'(n,d,e)$, for $1\leq e\leq 4$}
\label{A''}
\begin{tabular}{cc|c|l|l|l}
\hline
  & $e$ & 1 & 2 & 3 & 4 \\
$d$ & & & & & \\
\hline
\hline
1 & & \multicolumn{4}{c}{$\sum_{k=0}^e \binom{n}{k}$} \\
\hline
2 & & \multicolumn{4}{c}{$\sum_{k=0}^e \binom{n-1}{k}$} \\
\hline
3 & & 1 & $\left\lfloor \frac{n+1}{2}\right\rfloor$ & $\begin{cases}
\left\lfloor \frac{n+1}{3}\left\lfloor \frac{n}{2}\right\rfloor\right\rfloor,&\text{if $n\not\equiv 4\bmod{6}$,
$n\not\in\{1,3\}$} \\
\left\lfloor \frac{(n+1)n}{6}\right\rfloor -1,&\text{if $n\equiv 4\bmod{6}$} \\
2,&\text{if $n=3$} \\
1,&\text{if $n=1$}
\end{cases}$
&
\begin{minipage}{0.42\textwidth}
$\begin{cases}
\left\lfloor \frac{n+1}{4} \left\lfloor \frac{n}{3} \left\lfloor \frac{n-1}{2} \right\rfloor \right\rfloor \right\rfloor+1, &\text{if $n\not\equiv 5\bmod{6}$} \\
\left\lfloor \frac{n+1}{4} \left( \left\lfloor \frac{n}{3} \left\lfloor \frac{n-1}{2}\right\rfloor \right\rfloor -1\right)\right\rfloor+1,&\text{if $n\equiv 5\bmod{6}$,}
\end{cases}$ \\
except possibly for $n\in\{22$, $34$, $46$, $58$, $70$, $82$, $94$, $118$, $142$, $154$, $166$, $178$, $190$, $202$, $214$, $274$, $286$, $454$, $466$, $478$, $958\}$
\end{minipage}
 \\
\hline
4 & & & $\left\lfloor \frac{n}{2}\right\rfloor$ & $\begin{cases}
\left\lfloor \frac{n}{3}\left\lfloor \frac{n-1}{2}\right\rfloor\right\rfloor,&\text{if $n\not\equiv 5\bmod{6}$,
$n\not\in\{1,2,4\}$} \\
\left\lfloor \frac{n(n-1)}{6}\right\rfloor -1,&\text{if $n\equiv 5\bmod{6}$} \\
2,&\text{if $n=4$} \\
1,&\text{if $n\in\{1,2\}$} \\
\end{cases}$
&
\begin{minipage}{0.42\textwidth}
$\begin{cases}
\left\lfloor \frac{n}{4} \left\lfloor \frac{n-1}{3} \left\lfloor \frac{n-2}{2} \right\rfloor \right\rfloor \right\rfloor+1, &\text{if $n\not\equiv 0\bmod{6}$} \\
\left\lfloor \frac{n}{4} \left( \left\lfloor \frac{n-1}{3} \left\lfloor \frac{n-2}{2}\right\rfloor \right\rfloor -1\right)\right\rfloor+1,&\text{if $n\equiv 0\bmod{6}$,}
\end{cases}$ \\
except possibly for $n\in\{23$, $35$, $47$, $59$, $71$, $83$, $95$, $119$, $143$, $155$, $167$, $179$, $191$, $203$, $215$, $275$, $287$, $455$, $467$, $479$, $959\}$
\end{minipage}
\\
\hline
5 & & & 1 & $\left\lfloor \frac{n+1}{3}\right\rfloor$ & $\begin{cases}
\left\lfloor \frac{(n+1)n}{12}\right\rfloor-1,&\text{if $n\equiv 6,9\bmod{12}$, $n\not\in\{9,18\}$} \\
\left\lfloor \frac{n+1}{4}\left\lfloor \frac{n}{3}\right\rfloor\right\rfloor,
&\text{if $n\not\equiv 6,9\bmod{12}$, $n\not\in\{7,8,10,16\}$}\\
25,&\text{if $n=18$} \\
\left\lfloor \frac{n+1}{4}\left\lfloor \frac{n}{3}\right\rfloor\right\rfloor -1,&\text{if $n\in\{8,16\}$} \\
\left\lfloor \frac{n+1}{4}\left\lfloor \frac{n}{3}\right\rfloor\right\rfloor -2,&\text{if $n\in\{7,9,10\}$} 
\end{cases}$
\\
\hline
6 & & & & $\left\lfloor \frac{n}{3}\right\rfloor$ & $\begin{cases}
\left\lfloor \frac{n(n-1)}{12}\right\rfloor-1,&\text{if $n\equiv 7,10\bmod{12}$, $n\not\in\{10,19\}$} \\
\left\lfloor \frac{n}{4}\left\lfloor \frac{n-1}{3}\right\rfloor\right\rfloor,
&\text{if $n\not\equiv 7,10\bmod{12}$, $n\not\in\{8,9,11,17\}$}\\
25,&\text{if $n=19$} \\
\left\lfloor \frac{n}{4}\left\lfloor \frac{n-1}{3}\right\rfloor\right\rfloor -1,&\text{if $n\in\{9,17\}$} \\
\left\lfloor \frac{n}{4}\left\lfloor \frac{n-1}{3}\right\rfloor\right\rfloor -2,&\text{if $n\in\{8,10,11\}$} 
\end{cases}$
\\
\hline
7 & & & & 1 & $\left\lfloor \frac{n+1}{4}\right\rfloor$ \\
\hline
8 & & & & & $\left\lfloor \frac{n}{4}\right\rfloor$ \\
\hline
9 & & & & & 1 \\
\hline
\end{tabular}
\end{sidewaystable}

We review some coding-theoretic terminology and notations next.

\subsection{Preliminaries}

The set of integers $\{1,\ldots,n\}$ is denoted by $[n]$. 

Let $\bbF_2^n$ be the vector space of all the binary $n$-tuples, endowed with the
{\em Hamming metric}. Specifically, the {\em Hamming distance}
$\Delta(\vu,\vv)$ between $\vu,\vv\in\bbF_2^n$ is defined as the number of 
positions where $\vu$ and $\vv$ differ. The {\em Hamming weight}
$\wt(\vu)$ of $\vu\in\bbF_2^n$ is its distance from the origin, that is, $\wt(\vu)=\Delta(\vu,\vzero)$.
For $\vu\in\bbF_2^n$ and $i\in [n]$, $\vu_i$ denotes the $i$th component of $\vu$.
The {\em support} $\supp(\vu)$ of $\vu\in\bbF_2^n$ is the set of positions of $\vu$ with nonzero
value, that is, $\supp(\vu)=\{i\in[n]: \vu_i=1\}$.

A {\em binary code} of length $n$
is a nonempty subset of $\bbF_2^n$ and its elements are called {\em codewords}.
Since we are concerned with only binary codes in this paper,
henceforth we omit the ``binary'' quantifier throughout. The
number of codewords in a code is called its {\em size}.
The {\em minimum distance}
of a code $\C$, denoted $\dist(\C)$, is 
the quantity $\min_{\vu,\vv\in\C,\vu\not=\vv}\{\Delta(\vu,\vv)\}$.
A code of length $n$ and minimum distance $d$ is denoted an $(n,d)$ code.
Given a code $\C\subseteq\bbF_2^n$, the translate of $\C$ by $\vu\in\bbF_2^n$ is
the code $\C+\vu = \{\vv+\vu:\vv\in\C\}$, where vector addition is in $\bbF_2^n$.

A {\em set system} is a pair $\S=(X,\A)$, where $X$ is a finite set of {\em points}
and $\A\subseteq 2^X$. Elements of $\A$ are called {\em blocks}. The {\em order}
of $\S$ is the number of points, $|X|$. The {\em size} of $\S$ is the number of blocks
in $\A$.
The natural bijection
between $\bbF_2^n$ and $2^{[n]}$ (where a vector $\vu\in\bbF_2^n$
corresponds to the set $\supp(\vu)\in 2^{[n]}$) implies that a binary code $\C$ of
length $n$ can be represented by a set system $\S$ of order $n$, and vice versa:
\begin{equation*}
\C \subseteq \bbF_2^n \longleftrightarrow ([n],\{\supp(\vu):\vu\in\C\}).
\end{equation*}
Hence, we may speak of the set system of a binary code. When more natural,
we deal with the set system of a binary code, rather than the binary code itself.

The {\em Hamming ball} of radius $r$ around $\vu$ is the set
\begin{equation*}
B(\vu,r)=\{\vv\in\bbF_2^n : \Delta(\vu,\vv)\leq r\}.
\end{equation*}
For a code $\C\subseteq \bbF_2^n$,
the quantity $e=\left\lfloor (\dist(\C)-1)/2\right\rfloor$ is referred to as the
{\em error-correction bound} of $\C$ \cite{Sudan:1997}. This terminology reflects that
given any $\vu\in\bbF_2^n$,
\begin{equation*}
| B(\vu,e) \cap \C | \leq 1,
\end{equation*}
so that any transmitted codeword corrupted by at most $e$ errors can be unambiguously decoded
(to the nearest codeword) with maximum likelihood decoding. If the number of errors is beyond
the error-correction bound, we may not always be able to decode to a unique codeword. However,
it is desirable that in this case, the decoding algorithm outputs a list of candidate codewords,
containing the transmitted codeword. This
motivates the definition of list decodable codes.

For positive integers $e$ and $\ell$, a code $\C\subseteq\bbF_2^n$ is {\em $(e,\ell)$-list decodable}
if every ball of radius $e$ contains at most $\ell$ codewords, that is,
\begin{equation*}
|B(\vu,e)\cap \C| \leq \ell, \text{ for all $\vu\in\bbF_2^n$}.
\end{equation*}

\section{The Function $A'(n,d,e)$}

The key function we study in this paper is $A'(n,d,e)$, defined to be the
maximum size of a set $S\subseteq\bbF_2^n$ contained in a ball of radius $e$ such that
$\Delta(\vu,\vv)\geq d$ for all distinct $\vu,\vv\in S$.
More formally,
\begin{equation}
\label{A'def}
A'(n,d,e) = \max\{ |S| : \text{ $S$ is an $(n,d)$ code, and $S\subseteq B(\vx,e)$ for some $\vx\in\bbF_2^n$} \}.
\end{equation}
Notice that translating a code does not affect its distance properties. Hence we may assume that
the maximum in (\ref{A'def}) is attained when $\vx=\vzero$.
The definition of $A'(n,d,e)$ can then take the following equivalent form:
\begin{equation*}
A'(n,d,e) = \max\{ |S| : \text{ $S$ is an $(n,d)$ code, and $\wt(\vu)\leq e$ for all $\vu\in S$} \}. \\
\end{equation*}
We call an $(n,d)$ code having codewords of weight at most $e$ an
$(n,d,e)$ {\em bounded-weight code}. Hence, determining $A'(n,d,e)$ is
equivalent to determining the maximum size of an $(n,d,e)$ bounded-weight code.
An $(n,d,e)$ bounded-weight code of size $A'(n,d,e)$ is said to be {\em optimal}.

In contrast, an $(n,d,e)$ {\em constant-weight code} is an $(n,d)$ code whose codewords are
all of weight $e$. The maximum size of an $(n,d,e)$ constant-weight code is denoted
$A(n,d,e)$. The determination of $A(n,d,e)$ has been a central problem in coding theory, with
a rich literature (see, for example, \cite{Brouweretal:1990}). 
The importance
of the function $A'(n,d,e)$ is only realized relatively recently \cite{Guruswami:2004},
due to the following observation.

\begin{proposition} \hfill
\label{A'}
\begin{enumerate}[(i)]
\item If $\ell\geq A'(n,d,e)$, then every $(n,d)$ code is $(e,\ell)$-list decodable.
\item If $\ell<A'(n,d,e)$, then there exists an $(n,d)$ code that is not $(e,\ell)$-list decodable.
\end{enumerate}
\end{proposition}

\begin{proof}
Suppose $\C$ is an $(n,d)$ code. Then $|B(\vu,e)\cap \C| \leq A'(n,d,e)$ for all $\vu\in\bbF_2^n$.
It follows that if $\ell\geq A'(n,d,e)$, then $\C$ is $(e,\ell)$-list decodable. If $\ell < A'(n,d,e)$,
then an $(n,d,e)$ bounded-weight code of size $A'(n,d,e)$ is an $(n,d)$ code that is not
$(e,\ell)$-list decodable.
\qed
\end{proof}

Consequently, the problem of determining $A'(n,d,e)$ has attracted some direct attention
\cite{CassutoBruck:2004,Elias:1991,Guruswami:2004,GuruswamiSudan:2001}.
The determination of the exact value of $A'(n,d,e)$ is no doubt a difficult problem, so
most work has gone to establishing upper bounds on $A'(n,d,e)$. Proposition \ref{A'}(i)
can be applied when an upper bound of $A'(n,d,e)$ is known.
However,
this is not true for Proposition \ref{A'}(ii). So we do not get as strong a conclusion as if we
know the exact value of $A'(n,d,e)$. A useful result proven by Elias \cite{Elias:1991} is
the following:

\begin{proposition}[Elias {\cite[Proposition 10(c)]{Elias:1991}}]
\label{Elias}
If $d$ is odd, then $A'(n,d,e)=A'(n+1,d+1,e)$.
\end{proposition}

In subsequent sections, we determine the value
of $A'(n,d,e)$ for several parameter sets. We end this section with some easy exact values.

\begin{proposition}
\label{easy}
\hfill
\begin{enumerate}[(i)]
\item If $d\geq 2e+1$, then $A'(n,d,e)=1$.
\item If $d=2e$, then $A'(n,d,e)=\left\lfloor n/e\right\rfloor$.
\item If $d=2e-1$, then $A'(n,d,e)=\left\lfloor (n+1)/e\right\rfloor$.
\item If $d=2$, then $A'(n,d,e)=\sum_{k=0}^e \binom{n-1}{k}$.
\item If $d=1$, then $A'(n,d,e)=\sum_{k=0}^e \binom{n}{k}$.
\end{enumerate}
\end{proposition}

\begin{proof}
\begin{enumerate}[(i)]
\item Since $\Delta(\vu,\vv)\leq \wt(\vu)+\wt(\vv)\leq 2e$, there can be at most one codeword in the code.
\item Follows from the observation that all the 
codewords must have weight $e$ and have disjoint supports.
\item Follows from part (ii) and Proposition \ref{Elias}.
\item Follows from part (v) and Proposition \ref{Elias}.
\item Taking all binary $n$-tuples of weight $e$ or less gives the required code of distance one.
\qed
\end{enumerate} 
\end{proof}

Proposition \ref{easy} can be used to completely determine the exact value of $A'(n,d,1)$.

\begin{corollary}
\label{A'(n,d,1)}
\begin{equation*}
A'(n,d,1)=
\begin{cases}
1,&\text{if $d\geq 3$} \\
n,&\text{if $d=2$} \\
n+1,&\text{if $d= 1$.}
\end{cases}
\end{equation*}
\end{corollary}

\section{Determining $A'(n,2e-2,e)$ and $A'(n,2e-3,e)$}

By Proposition \ref{Elias}, we have $A'(n,2e-2,e)=A'(n-1,2e-3,e)$, so it suffices to focus
on the case $d=2e-2$ in this section.
When $d=2e-2$, $A'(n,d,e)$ can be expressed in terms of $A(n,d,e)$.

%

\begin{proposition}
\label{d=2e-2}
When $d=2e-2$,
\begin{equation*}
A'(n,d,e) = 
\max_{\substack{\alpha\in\{0,1\} \\ 0\leq\beta\leq \frac{n}{e-1} \\ \alpha\beta=0}}
\{\alpha+\beta+A(n-\alpha(e-2)-\beta(e-1),d,e)\}.
\end{equation*}
\end{proposition}

\begin{proof}
By considering the distance between codewords, we see that
an $(n,2e-2,e)$ bounded-weight code $\C$ must satisfy:
\begin{enumerate}[(i)]
\item $\wt(\vu)\geq e-2$ for all $\vu\in\C$;
\item there exists at most one $\vu\in\C$ such that $\wt(\vu)=e-2$;
\item if there exists $\vu\in\C$ such that $\wt(\vu)=e-2$, then there do not exist any $\vv\in\C$
such that $\wt(\vv)=e-1$.
\end{enumerate}
Let $\alpha\in\{0,1\}$ be the number of codewords in $\C$ of weight $e-2$,
and let $0\leq\beta\leq n/(e-1)$ be the number of codewords in $\C$ of weight $e-1$.
Note that $\alpha\beta=0$ by property (iii) above.
The supports of these $\alpha+\beta$ codewords of weight $e-2$ and $e-1$
are pairwise disjoint, and are also pairwise disjoint
from those of codewords of weight $e$. Hence, if we shorten $\C$ at the positions
containing nonzero values among these $\alpha+\beta$ codewords, we end up with
an $(n-\alpha(e-2)-\beta(e-1),2e-2,e)$ constant-weight code. It follows that
\begin{equation*}
A'(n,2e-2,e) = 
\max_{\substack{\alpha\in\{0,1\} \\ 0\leq\beta\leq \frac{n}{e-1} \\ \alpha\beta=0}}
\{\alpha+\beta+A(n-\alpha(e-2)-\beta(e-1),2e-2,e)\}. 
\end{equation*}
\qed
\end{proof}

Hence, for any fixed $e$, we can determine
$A'(n,2e-2,e)$ whenever the exact value of $A(n,2e-2,e)$ is known for all $n$. The following
are classical results from the theory of constant-weight codes.

\begin{theorem}[Sch\"onheim \cite{Schonheim:1966}, Spencer\cite{Spencer:1968},
Brouwer \cite{Brouwer:1979}]
\label{classicalcwc}
\hfill
\begin{enumerate}[(i)]
\item
\begin{equation*}
A(n,4,3)=\begin{cases}
\left\lfloor \frac{n}{3} \left\lfloor \frac{n-1}{2} \right\rfloor \right\rfloor-1,&\text{if $n\equiv 5\bmod{6}$}\\
\left\lfloor \frac{n}{3} \left\lfloor \frac{n-1}{2} \right\rfloor \right\rfloor,&\text{otherwise}.
\end{cases}
\end{equation*}
\item
\begin{equation*}
A(n,6,4)=\begin{cases}
\left\lfloor \frac{n}{4} \left\lfloor \frac{n-1}{3}\right\rfloor \right\rfloor-1,&\text{if $n\equiv 7,10\bmod{12}$ and $n\not\in\{10,19\}$} \\
\left\lfloor \frac{n}{4} \left\lfloor \frac{n-1}{3}\right\rfloor \right\rfloor-1,&\text{if $n\in\{9,17\}$}\\
\left\lfloor \frac{n}{4} \left\lfloor \frac{n-1}{3}\right\rfloor \right\rfloor-2,&\text{if $n\in\{8,10,11\}$}\\
\left\lfloor \frac{n}{4} \left\lfloor \frac{n-1}{3}\right\rfloor \right\rfloor-3,&\text{if $n=19$} \\
\left\lfloor \frac{n}{4} \left\lfloor \frac{n-1}{3}\right\rfloor \right\rfloor,&\text{otherwise}. \\
\end{cases}
\end{equation*}
\end{enumerate}
\end{theorem}

Note that $A(n,d-1,e)=A(n,d,e)$ when $d$ is even, so results on $A(n,d,e)$ are normally
stated only for even $d$.

%

\begin{corollary}
\label{A'(n,4,3)}
For $n$ a positive integer, $A'(n,4,3)=A(n,4,3)$, except
when $n\in\{1,2,4\}$, in which case, we have
$A'(1,4,3)=A'(2,4,3)=1$ and $A'(4,4,3)=2$.
\end{corollary}

\begin{proof}
The values of $A'(n,4,3)$ are trivial to obtain when $n\leq 4$, so we assume 
henceforth that $n\geq 5$.
From Proposition \ref{d=2e-2} and Theorem \ref{classicalcwc}(ii), we have
\begin{align*}
A'(n,4,3)&=\max_{\substack{\alpha\in\{0,1\} \\ 0\leq\beta\leq n/2 \\ \alpha\beta=0}}
\{\alpha+\beta+A(n-\alpha-2\beta,4,3)\} \\
&=\max_{0\leq\beta\leq n/2}
\{\beta+A(n-2\beta,4,3), 1+A(n-1,4,3)\}.
\end{align*}

First, note that since $n\geq 5$, we have $1+A(n-1,4,3)\leq A(n,4,3)$.

Now, take any $(n,4,3)$ bounded-weight code $\C$ attaining the bound
$\beta+A(n-2\beta,4,3)$, that is, $\C$ contains $\beta$ codewords of weight two
and then outside the supports of these codewords of weight two, $\C$ contains the
remaining codewords of weight three. Without loss of generality, assume
that the last coordinate position does not belong to any of the supports of the
codewords of weight two. We consider two cases:

\begin{description}
\item[$2\beta\not= n$:] Change the last coordinate (from zero to one)
in every codeword of weight two in $\C$
to obtain a new code $\C'$. Clearly, the distances between the codewords in $\C'$ obtained
from those of weight two in $\C$ are unaffected, and remain at least four.
The distances between codewords of weight two and codewords of weight three in $\C$ are
all five, so the distances between codewords in $\C'$ obtained from codewords of weight two in $\C$
and the other codewords in $\C$ are all at least four. Therefore, $\C'$ is an
$(n,4,3)$ constant-weight code, and hence its size can never exceed $A(n,4,3)$.
Since $|\C|=|\C'|$, we also have $|\C|\leq A(n,4,3)$. Furthermore, equality holds when
$\beta=0$.

\item[$2\beta=n$:] In this case, $\beta+A(n-2\beta,4,3)=n/2+A(0,4,3)\leq A(n,4,3)$ when
$n\geq 6$.
\end{description}
It follows that $A'(n,4,3)=A(n,4,3)$ when $n\geq 5$.
\qed
\end{proof}

We skip the proof for Corollary \ref{A'(n,6,4)} below as it is similar to that for
Corollary \ref{A'(n,4,3)}.

\begin{corollary}
\label{A'(n,6,4)}
$A'(n,6,4)=A(n,6,4)$ for all positive integers $n$, except
when $n\in\{1,2,3,6\}$, in which case, we have
$A'(n,6,4)=1$ for $n\in\{1,2,3\}$ and $A'(6,6,4)=2$.
\end{corollary}

We end this section by giving the exact value of $A'(n,d,2)$ and $A'(n,d,3)$.

\begin{theorem}
\label{A'(n,d,2)}
\begin{equation*}
A'(n,d,2)=
\begin{cases}
1,&\text{if $d\geq 5$} \\
\left\lfloor n/2\right\rfloor,&\text{if $d=4$} \\
\left\lfloor (n+1)/2\right\rfloor,&\text{if $d=3$} \\
\binom{n}{2}+1,&\text{if $d=2$} \\
\binom{n}{2}+n+1,&\text{if $d\leq 1$.}
\end{cases}
\end{equation*}
\end{theorem}

\begin{proof}
The value of $A'(n,d,2)$ follows from Proposition \ref{easy}.
\qed
\end{proof}

\begin{theorem}
\label{A'(n,d,3)}
When $(n,d)\not\in\{(1,3)$, $(1,4)$, $(2,4)$, $(3,3)$, $(4,4)\}$, we have
\begin{equation*}
A'(n,d,3)=
\begin{cases}
1,&\text{if $d\geq 7$} \\
\left\lfloor n/3\right\rfloor,&\text{if $d=6$} \\
\left\lfloor (n+1)/3\right\rfloor,&\text{if $d=5$} \\
\left\lfloor \frac{n}{3} \left\lfloor \frac{n-1}{2}\right\rfloor\right\rfloor,&\text{if $d=4$, $n\not\equiv 5\bmod{6}$, $n\not\in\{1,2,4\}$} \\
\left\lfloor \frac{n(n-1)}{6} \right\rfloor-1,&\text{if $d=4$, $n\equiv 5\bmod{6}$} \\
\left\lfloor \frac{n+1}{3} \left\lfloor \frac{n}{2}\right\rfloor\right\rfloor,&\text{if $d=3$, $n\not\equiv 4\bmod{6}$, $n\not\in\{1,3\}$} \\
\left\lfloor \frac{(n+1)n}{6} \right\rfloor-1,&\text{if $d=3$, $n\equiv 4\bmod{6}$} \\\sum_{k=0}^3\binom{n-1}{k},&\text{if $d=2$} \\
\sum_{k=0}^3\binom{n}{k},&\text{if $d\leq 1$.}
\end{cases}
\end{equation*}
The remaining values are given by
\begin{align*}
A'(1,3,3) = A'(1,4,3) = A'(2,4,3) = 1, \\
A'(3,3,3) = A'(4,4,3) = 2.
\end{align*}
\end{theorem}

\begin{proof}
The value of $A'(n,d,3)$ for $d\geq 5$ and $d\leq 2$ follows from Proposition \ref{easy}, and
for $d=4$ follows from Corollary \ref{A'(n,4,3)}. The remaining required value of
$A'(n,3,3)$ follows by applying Proposition \ref{Elias} to Corollary \ref{A'(n,4,3)}.
\qed
\end{proof}

\section{Determining $A'(n,4,4)$}

We consider the case $e=4$ and $d=4$ in this section,
since $A'(n,d,e)$ has already been determined for all $e\leq 3$ in
Corollary \ref{A'(n,d,1)}, Theorem \ref{A'(n,d,2)} and Theorem \ref{A'(n,d,3)},
and for $e=4$ and $d=6$ (and hence also $d=5$) in Corollary \ref{A'(n,6,4)}.
The main result in this section is that $A'(n,4,4)=A(n,4,4)+1$ for all sufficiently large $n$.

\begin{table}
\renewcommand{\arraystretch}{1.3}
\caption{Elements of $\N$, the possible exceptions to $A(n,4,4)=J(n,4,4)$}
\label{exceptions}
\centering
\begin{tabular}{rrrrrrrrrrc}
\hline
$23$ & $35$ & $47$ & $59$ & $71$ & $83$ & $95$ & $119$ & $143$ & $155$ & $167$ \\
$179$ & $191$ & $203$ & $215$ & $275$ & $287$ & $455$ & $467$ & $479$ & $959$ \\
\hline
\end{tabular}
\end{table}

Let
\begin{equation*}
J(n,4,4) = 
\begin{cases}
\left\lfloor \frac{n}{4} \left\lfloor \frac{n-1}{3} \left\lfloor \frac{n-2}{2} \right\rfloor \right\rfloor \right\rfloor, &\text{if $n\not\equiv 0\bmod{6}$} \\
\left\lfloor \frac{n}{4} \left( \left\lfloor \frac{n-1}{3} \left\lfloor \frac{n-2}{2}\right\rfloor \right\rfloor -1\right)\right\rfloor,&\text{if $n\equiv 0\bmod{6}$.}
\end{cases}
\end{equation*}
The following bound was obtained by Johnson \cite{Johnson:1972}:
\begin{equation}
A(n,4,4) \leq J(n,4,4).
\end{equation}
With recent advances, the value of $A(n,4,4)$ is now known for all but  the values of $n\in\N$,
where $\N$ is the set of 21 numbers in Table \ref{exceptions}.

\begin{theorem}[Hanani \cite{Hanani:1960}, Brouwer \cite{Brouwer:1978}, Ji \cite{Ji:2006}]
\label{Ji}
$A(n,4,4)=J(n,4,4)$ for all positive integers $n$, except possibly for
$n\in\N$.
\end{theorem}

Throughout this section, $\C$ is an $(n,4,4)$ bounded-weight code and
$\S=([n],\A)$ is the set system of $\C$. Note that 
taking an optimal $(n,4,4)$ constant-weight code together with the
zero-weight codeword yields an $(n,4,4)$ bounded-weight code of size $A(n,4,4)+1$.
Therefore,
\begin{equation}
\label{lb}
A'(n,4,4) \geq A(n,4,4)+1 = J(n,4,4)+1.
\end{equation}
To establish that the inequality in (\ref{lb}) is indeed an equality,
we show below that $A'(n,4,4) \leq J(n,4,4)+1$.
We do this by case analysis on the number of codewords of weight $k$ in $\C$,
and counting.

Let $m_k$, $k\in\{0,1,2,3,4\}$, 
denote the number of blocks of size $k$ in $\A$.
If $\A$ contains the empty set as a block, then all other blocks of $\A$ must have
size four, and hence the size of $\S$ in this case is at most $A(n,4,4)+1$.
Henceforth, assume $m_0=0$.

Clearly, $m_1\in\{0,1\}$, since $\C$ has minimum distance four.
If $m_1=1$ and $A=\{x\}$ is a block in $\A$,
then $x$ cannot be contained in any other block of $\A$, and
moreover all other blocks of $\A$ must have size three or four. Hence,
it follows that
the size of $\S$ in this case is at most one more than the largest size of an
$(n-1,4,4)$ bounded-weight code with only codewords of size three and four.
We can therefore restrict our attention to 
the case when $m_1=0$. 

For $T\subseteq [n]$ and $k\in\{2,3,4\}$, let
$f_k(T)$ denote the number of blocks of size $k$ in $\A$ that
contain $T$. For succinctness of notation, we suppress braces in the argument
of $f_k$ so that, for example, we write $f_k(x)$ for $f_k(\{x\})$ and $f_k(x,y)$
for $f_k(\{x,y\})$.
Define, for $k\in\{2,3\}$,
\begin{equation*}
D_k=\{x\in[n]: f_k(x) > 0\}.
\end{equation*}
Then $D_2\cap D_3=\emptyset$, since $\dist(\C)=4$. Let $D_4=[n]\setminus(D_2\cup D_3)$.

The following equations are obtained by counting points in the blocks of $\A$:
\begin{align}
\label{kmk1}
\sum_{x\in D_2} f_2(x) &= 2m_2, \\
\label{kmk2}
\sum_{x\in D_3} f_3(x) &= 3m_3, \\
\label{kmk3}
\sum_{x\in [n]} f_4(x) &= 4m_4.
\end{align}
We also have the inequality
\begin{equation*}
f_2(x,y)+f_3(x,y)\leq 1,
\end{equation*}
for $\{x,y\}\subseteq[n]$.
If $f_2(x,y)+f_3(x,y)=1$, then $f_4(x,y)=0$. Otherwise, we have the following.

\begin{lemma}
\label{f(x,y)ub}
Let $\{x,y\}\subseteq[n]$.
If $f_2(x,y)+f_3(x,y)=0$, then
\begin{equation*}
f_4(x,y)\leq \begin{cases}
\left\lfloor \frac{n-4}{2}\right\rfloor,&\text{if $x\in D_2$ and $y\in D_2$} \\
\left\lfloor \frac{n-3-2f_3(y)}{2}\right\rfloor, &\text{if $x\in D_2$ and $y\in D_3$} \\
\left\lfloor \frac{n-3}{2} \right\rfloor, &\text{if $x\in D_2$ and $y\in D_4$} \\
\left\lfloor \frac{n-2-2f_3(x)-2f_3(y)}{2}\right\rfloor, &\text{if $x\in D_3$ and $y\in D_3$} \\
\left\lfloor \frac{n-2-2f_3(x)}{2}\right\rfloor, &\text{if $x\in D_3$ and $y\in D_4$} \\
\left\lfloor \frac{n-2}{2}\right\rfloor, &\text{if $x\in D_4$ and $y\in D_4$.}
\end{cases}
\end{equation*}
\end{lemma}

\begin{proof}
First note that the blocks of size four containing $\{x,y\}$ have pairwise
intersection exactly $\{x,y\}$, so that the number of such blocks is at most
$\left\lfloor (n-2-|E|)/2\right\rfloor$, where $E\subseteq [n]\setminus\{x,y\}$
is a set of points that must be excluded in any blocks of size four containing $\{x,y\}$.
\begin{enumerate}[(i)]
\item When $x\in D_2$ and $y\in D_2$, suppose that the two blocks of size two containing $x$
and $y$, respectively, are $\{x,a\}$ and $\{y,b\}$. Then taking $E=\{a,b\}$ gives
$f_4(x,y)\leq \left\lfloor (n-4)/2\right\rfloor$.

\item When $x\in D_2$ and $y\in D_3$, suppose that the block of size two
containing $x$ is $\{x,a\}$, and the blocks of size three containing $y$ are
$\{y,b_i,c_i\}$, $i=1,\ldots, f_3(y)$. Then taking $E=\{a,b_i,c_i: i=1,\ldots,f_3(y)\}$
gives 
\begin{equation*}
f_4(x,y)\leq \left\lfloor (n-3-2f_3(y))/2\right\rfloor.
\end{equation*}

\item When $x\in D_2$ and $y\in D_4$, suppose that the block of size two containing $x$
is $\{x,a\}$. Then taking $E=\{a\}$ gives
$f_4(x,y)\leq \left\lfloor (n-3)/2\right\rfloor$.

\item When $x\in D_3$ and $y\in D_3$, suppose that the blocks of size three
containing $x$ and the blocks of size three containing $y$ are
$\{x,a_i,b_i\}$, $i=1,\ldots,f_3(x)$, and $\{y,c_i,d_i\}$, $i=1,\ldots,f_3(y)$, respectively.
Then taking $E=\{a_i,b_i: i=1,\ldots,f_3(x)\}\cup\{c_i,d_i: i=1,\ldots,f_3(y)\}$ gives
\begin{equation*}
f_4(x,y)\leq\left\lfloor (n-2-2f_3(x)-2f_3(y))/2\right\rfloor.
\end{equation*}

\item When $x\in D_3$ and $y\in D_4$, suppose that the blocks of size three
containing $x$ are $\{x,a_i,b_i\}$, $i=1,\ldots,f_3(x)$. Then taking
$E=\{a_i,b_i:i=1,\ldots,f_3(x)\}$ gives
$f_4(x,y)\leq \left\lfloor (n-2-2f_3(x))/2\right\rfloor$.

\item When $x\in D_4$ and $y\in D_4$, taking $E=\emptyset$ gives
$f_4(x,y)\leq \left\lfloor (n-2)/2\right\rfloor$. \qed
\end{enumerate}
\end{proof}

Counting in two ways the number of 2-subsets $T\subseteq[n]$ containing the point $x$
such that $T$ is contained in a block of size four gives
\begin{equation}
\label{3f4}
3 f_4(x) = \sum_{y\in [n]\setminus\{x\}} f_4(x,y),
\end{equation}
since each block of size four contains three
2-subsets, each of which contains $x$.

Applying the inequality on $f_4(x,y)$ in Lemma \ref{f(x,y)ub} to (\ref{3f4}), and
recalling that if $f_2(x,y)+f_3(x,y)=1$ then $f_4(x,y)=0$, gives
the inequality
{
\begin{align}
\label{3f4*}
&3f_4(x) \leq \nonumber \\
&\begin{cases}
(|D_2|-2)\left\lfloor \frac{n-4}{2} \right\rfloor + |D_3|\left\lfloor \frac{n-5}{2}\right\rfloor + 
(n-|D_2|-|D_3|) \left\lfloor \frac{n-3}{2}\right\rfloor, & \text{if $x\in D_2$} \\
|D_2|\left\lfloor \frac{n-5}{2}\right\rfloor + (|D_3|-1-2f_3(x))\left\lfloor \frac{n-6}{2}\right\rfloor + 
(n-|D_2|-|D_3|)\left\lfloor \frac{n-4}{2}\right\rfloor, & \text{if $x\in D_3$} \\
|D_2|\left\lfloor \frac{n-3}{2}\right\rfloor + |D_3|\left\lfloor \frac{n-4}{2}\right\rfloor +
(n-1-|D_2|-|D_3|)\left\lfloor \frac{n-2}{2}\right\rfloor, & \text{if $x\in D_4$.}
\end{cases}
\end{align}
}

We are now ready to establish upper bounds on the size of $\S$ when $m_0=m_1=0$.
By (\ref{kmk1})--(\ref{kmk3}),
\begin{align*}
|\A| &= m_2+m_3+m_4 \nonumber \\
& =
\frac{1}{2} \sum_{x\in D_2} f_2(x) + \frac{1}{3}\sum_{x\in D_3}f_3(x) +
\frac{1}{4}\sum_{x\in [n]} f_4(x)  \nonumber \\
&= \frac{1}{2} \sum_{x\in D_2} f_2(x) + \frac{1}{3}\sum_{x\in D_3}f_3(x) + \frac{1}{4}\left( \sum_{x\in D_2} f_4(x) + \sum_{x\in D_3} f_4(x) + \sum_{x\in D_4} f_4(x)\right) \nonumber \\
&= \frac{1}{4} \left(
\sum_{x\in D_2} (2f_2(x) + f_4(x))
 + \sum_{x\in D_3} \left(\frac{4}{3} f_3(x) + f_4(x)\right) + 
\sum_{x\in D_4} f_4(x)
 \right).
\end{align*}
Let 
\begin{align*}
F_2(x) &= 2f_2(x)+f_4(x), \\
F_3(x) &= \frac{4}{3}f_3(x) + f_4(x), \\
F_4(x) & = f_4(x),
\end{align*}
so that 
\begin{equation}
\label{|A|}
|\A|=\frac{1}{4}\left(\sum_{x\in D_2} F_2(x) + \sum_{x\in D_3} F_3(x) + \sum_{x\in D_4} F_4(x)
\right).
\end{equation}

\subsection{$n\equiv 1\bmod{2}$}
\label{nodd}

In this subsection, we consider the case when $n$ is odd.

An upper bound on $F_2(x)=2f_2(x)+f_4(x)$ can be obtained by observing
that there can be at most one block of size two containing $x$ (and hence $f_2(x)\leq 1$),
and applying (\ref{3f4*}) to upper bound $f_4(x)$. More specifically, when
$x\in D_2$, we have
\begin{align*}
F_2(x) &= 2f_2(x)+f_4(x) \\
&\leq 2+\frac{1}{3}\left(
(|D_2|-2)\left\lfloor \frac{n-4}{2} \right\rfloor + |D_3|\left\lfloor \frac{n-5}{2}\right\rfloor + 
(n-|D_2|-|D_3|) \left\lfloor \frac{n-3}{2}\right\rfloor
\right) \\
&= 2+\frac{1}{3}\left(
(|D_2|-2)\cdot \frac{n-5}{2} + |D_3|\cdot\frac{n-5}{2} + 
(n-|D_2|-|D_3|) \cdot\frac{n-3}{2}
\right) \\
&= \frac{n^2-5n-2D_2-2D_3-22}{6} \\
&= \frac{n-1}{3}\cdot\frac{n-3}{2} - \frac{n+2|D_2|+2|D_3|-19}{6}.
\end{align*}
Since $F_2(x)$ is an integer, we have
\begin{equation*}
F_2(x) \leq \left\lfloor \frac{n-1}{3}\cdot\frac{n-3}{2} - \frac{n+2|D_2|+2|D_3|-19}{6}\right\rfloor,
\ \ \ \ \ \text{when $x\in D_2$}.
\end{equation*}
We can similarly derive
\begin{align*}
F_3(x) &\leq \left\lfloor \frac{n-1}{3}\cdot  \frac{n-3}{2} - \nonumber  \frac{n+|D_3|-2+(n-11)f_3(x)}{3} \right\rfloor, \ \ \ \ \ \text{when $x\in D_3$},\\
F_4(x) &\leq \left\lfloor \frac{n-1}{3}\cdot\frac{n-3}{2} - \frac{|D_3|}{3}\right\rfloor,
\ \ \ \ \ \text{when $x\in D_4$}.
\end{align*}

If $|D_2|\not=0$ (that is $|D_2|\geq 2$), then 
$\frac{n+2|D_2|+2|D_3|-19}{6}\geq -2/3$ when $n\geq 11$,
so that $F_2(x)\leq \left\lfloor \frac{n-1}{3}\cdot\frac{n-3}{2}\right\rfloor$ for $n\geq 11$.

If $|D_3|\not=0$ and $x\in D_3$, then $|D_3| \geq 2f_3(x)+1$ since in each block of
size three, $x$ appears with two other points, and these points are all distinct. In this case,
$n+|D_3|-2+(n-11)f_3(x)\geq 0$ when $n\geq 9$, so that
$F_3(x) \leq \left\lfloor \frac{n-1}{3}\cdot\frac{n-3}{2}\right\rfloor$ for $n\geq 9$.

Hence, when $|D_2|\not=0$ or $|D_3|\not=0$,
each of $F_2(x)$, $F_3(x)$, and $F_4(x)$ is at most 
$\left\lfloor \frac{n-1}{3}\cdot\frac{n-3}{2}\right\rfloor$ for $n\geq 11$. It
follows from (\ref{|A|}) that
\begin{align*}
|\A| &\leq \frac{1}{4}\sum_{x\in [n]}\left\lfloor \frac{n-1}{3}\cdot\frac{n-3}{2}\right\rfloor \\
&= \left\lfloor \frac{n}{4} \left\lfloor \frac{n-1}{3} \left\lfloor \frac{n-2}{2}\right\rfloor \right\rfloor 
\right\rfloor = J(n,4,4).
\end{align*}
We summarize the results in this section as:

\begin{proposition}
\label{odd}
When $n\equiv 1\bmod{2}$, $n\geq 11$, an optimal $(n,4,4)$ bounded-weight code $\C$
has $|D_2|=|D_3|=0$, that is, $\C$ contains no codewords of weight two and weight three,
except possibly for $n\in\{23$, $35$, $47$, $59$, $71$, $83$, $95$, $119$, $143$, $155$, $167$, $179$, $191$, $203$, $215$, $275$, $287$, $455$, $467$, $479$, $959\}$.
\end{proposition}

\subsection{$n\equiv 2$ or $4\bmod{6}$}

In this subsection, we consider the case when $n\equiv 2$ or $4\bmod{6}$.

Using (\ref{3f4*}) and simplifying gives, when $x\in D_2$:
\begin{equation*}
F_2(x) \leq \left\lfloor \frac{n-1}{3}\cdot\frac{n-2}{2} - \frac{3n+2|D_3|-18}{6}\right\rfloor,
\end{equation*}
when $x\in D_3$:
\begin{equation*}
F_3(x) \leq \left\lfloor \frac{n-1}{3}\cdot \frac{n-2}{2} -
 \frac{n+|D_2|+|D_3|-2+(n-10)f_3(x)}{3} \right\rfloor,
\end{equation*}
and when $x\in D_4$:
\begin{equation*}
F_4(x) \leq \left\lfloor \frac{n-1}{3}\cdot\frac{n-2}{2} - \frac{|D_2|+|D_3|}{3}\right\rfloor.
\end{equation*}

If $|D_2|\not=0$, then since $3n+2|D_3|-18\geq 0$ when $n\geq 6$, we have 
$F_2(x)\leq \left\lfloor \frac{n-1}{3}\cdot\frac{n-2}{2}\right\rfloor$ for $n\geq 6$.

If $|D_3|\not=0$, then $|D_3|\geq 2f_3(x)+1$, giving
$n+|D_2|+|D_3|-2+(n-10)f_3(x)\geq 0$ when $n\geq 8$, so that
$F_3(x)\leq \left\lfloor \frac{n-1}{3}\cdot\frac{n-2}{2}\right\rfloor$ for $n\geq 8$.

As in subsection \ref{nodd}, we deduce the following:

\begin{proposition}
\label{24}
When $n\equiv 2$ or $4\bmod{6}$, $n\geq 8$, an optimal $(n,4,4)$ bounded-weight code
$\C$ has $|D_2|=|D_3|=0$, that is, $\C$ contains no codewords of weight two and weight three.
\end{proposition}

\subsection{$n\equiv 0\bmod{6}$}

Here, the remaining case of $n\equiv 0\bmod{6}$ is addressed.
First note that
\begin{equation*}
\left\lfloor \frac{n-1}{3} \left\lfloor \frac{n-2}{2}\right\rfloor \right\rfloor -1 =
\frac{n^2-3n-6}{6}.
\end{equation*}

Using (\ref{3f4*}) and simplifying gives, when $x\in D_2$:
\begin{equation*}
F_2(x) \leq \left\lfloor \frac{n^2-3n-6}{6} - \frac{3n+2|D_3|-26}{6}\right\rfloor,
\end{equation*}
when $x\in D_3$:
\begin{equation*}
F_3(x) \leq \left\lfloor \frac{n^2-3n-6}{6} - 
 \frac{n+|D_2|+|D_3|-6+(n-10)f_3(x)}{3} \right\rfloor,
\end{equation*}
and when $x\in D_4$:
\begin{equation*}
F_4(x) \leq \left\lfloor \frac{n^2-3n-6}{6}- \frac{|D_2|+|D_3|-4}{3}\right\rfloor.
\end{equation*}

If $|D_2|\not=0$, then since $3n+2|D_3|-26\geq 0$ when $n\geq 9$, we have 
$F_2(x)\leq \frac{n^2-3n-6}{6}$ for $n\geq 9$.

If $|D_3|\not=0$, then $|D_3|\geq 2f_3(x)+1$, giving
$n+|D_2|+|D_3|-6+(n-10)f_3(x)\geq 0$ when $n\geq 8$, so that
$F_3(x)\leq \frac{n^2-3n-6}{6}$ for $n\geq 8$.

If $|D_2\cup D_3|\not=0$, then $|D_2\cup D_3|\geq 2$, and hence 
$F_4(x)\leq \left\lfloor \frac{n^2-3n-6}{6}- \frac{|D_2|+|D_3|-4}{3}\right\rfloor\leq 
\frac{n^2-3n-6}{6}$.

We therefore have:

\begin{proposition}
\label{0}
When $n\equiv 0\bmod{6}$, $n\geq 12$, an optimal $(n,4,4)$ bounded-weight code
$\C$ has $|D_2|=|D_3|=0$, that is, $\C$ contains no codewords of weight two and weight three.
\end{proposition}

\subsection{Optimal $(n,4,4)$ Bounded-Weight Codes of Small Lengths}
\label{small}

\begin{table}
\renewcommand{\arraystretch}{1.3}
\caption{Some Optimal $(n,4,4)$ Bounded-Weight Codes}
\label{small4}
\centering
\begin{tabular}{r | l}
\hline
$n$ & Codewords \\
\hline\hline
6 & 111000 100101 010110 001011 \\
\hline
7 & 0000000
1110010
1101001
1010101
1001110
0111100
0100111
0011011 \\
\hline
9 & 000000000
111010000
110100100
110001001
101100001 
101000110 
100101010 \\
& 100011100
100010011
011100010 
011001100
010110001 
010011010
010000111 \\
& 001111000 
001010101
001001011
000110110 
000101101 \\
\hline
\end{tabular}
\end{table}

The values of $A'(n,4,4)$
for several small values of $n$ are provided below.

\begin{proposition}
\label{smallcodes}
\begin{align*}
A'(n,4,4) &=
\begin{cases}
1,&\text{if $n\leq 3$} \\
2,&\text{if $n\in\{4,5\}$} \\
4,&\text{if $n=6$} \\
8,&\text{if $n=7$} \\
19,&\text{if $n=9$.} \\
\end{cases}
\end{align*}
\end{proposition}

\begin{proof}
The value of $A'(n,4,4)$ is easily obtained for $n\leq 5$.
The optimal $(n,4,4)$ bounded-weight codes for $n\in\{6,7,9\}$
are given in Table \ref{small4}. The codes are obtained via exhaustive search.
\qed
\end{proof}

Proposition \ref{smallcodes} shows that $A'(n,4,4)=A(n,4,4)+1$ for $n\in\{3,4,5,6,7,9\}$.

\subsection{Piecing Together}

Let $n\not\in \N\cup\{1,2,3,4,5,6,7,9\}$, and let $\C$
be an $(n,4,4)$ bounded-weight code.

Propositions \ref{odd}, \ref{24}, and \ref{0}
imply that $\C$ has size
at most $J(n,4,4)$ if $\C$ contains codewords of weight two and/or three.
Since (\ref{lb}) with Theorem \ref{Ji} gives
$A'(n,4,4)\geq J(n,4,4)+1$, $\C$ cannot be optimal.
Therefore, if $\C$ is optimal, $\C$ can
contain only codewords of weight zero, one, or four.

Suppose $\C$ contains only codewords of weight zero and four.
Then obviously, $|\C| \leq A(n,4,4)+1=J(n,4,4)+1$.
If $\C$ contains a codeword of weight one,
we have seen earlier that $\C$ has size at most one more than the size of
the largest $(n-1,4,4)$ bounded weight code containing only codewords of weight three
and four. As shown above, such a code has size at most $J(n-1,4,4)$. Since
$J(n-1,4,4)+1\leq J(n,4,4)+1$, we may assume that if $\C$ is optimal, then $\C$ has
no codewords of weight one. With the results in subsection \ref{small}, we now have:

\begin{theorem}
For all positive integers $n$,
\begin{equation*}
A'(n,4,4) = J(n,4,4)+1,
\end{equation*}
except possibly for $n\in\N$.
\end{theorem}

\section{Conclusion}

In this paper, we continue the investigation of the function $A'(n,d,e)$, which gives the
smallest possible $\ell$ so that every $(n,d)$ code is list decodable with a list of length $\ell$ up to
radius $e$. Exact values of $A'(n,d,e)$ were determined for $d\geq 2e-5$ and $d\leq 3$.
As a result, the exact value of $A'(n,d,e)$ is now known for all but 42
values of $n$, when $e\leq 4$.
Our approach in this paper is purely combinatorial, and we have not attempted to address
the existence of codes admitting efficient list decoding algorithms
capable of meeting the bounds established here.

\begin{acknowledgements}
We thank the anonymous reviewer for helpful comments.
\end{acknowledgements}

%
%
%

\bibliographystyle{spmpsci}      
\bibliography{/Users/ymchee/Documents/Bibliography/mybibliography}   

\end{document}